\documentclass[conference]{IEEEtran}
\IEEEoverridecommandlockouts
\usepackage{amsmath,amssymb,amsfonts}
\usepackage{amsthm}
\usepackage[export]{adjustbox}
\usepackage{algorithm}
\usepackage{algpseudocode}
\usepackage{graphicx}
\newtheorem{theorem}{Theorem}
\newtheorem{corollary}{Corollary}
\usepackage{textcomp}
\usepackage{xcolor}
\def\BibTeX{{\rm B\kern-.05em{\sc i\kern-.025em b}\kern-.08em
T\kern-.1667em\lower.7ex\hbox{E}\kern-.125emX}}

\usepackage[style=ieee, doi=false, url=false, isbn=false, giveninits=true, backend=biber, sorting=none, sortcites=true, date=year]{biblatex}
\begin{document}

\title{Send Pilot or Data? Leveraging Age of Channel State Information for Throughput Maximization\\

\thanks{This work was supported in part by the NSF grant CNS-2239677.}
}

\author{\IEEEauthorblockN{Sirin Chakraborty}
\IEEEauthorblockA{\textit{Dept. of ECE} \\
\textit{Auburn University}\\
Auburn, AL, USA \\
szc0260@auburn.edu}
\and
\and
\IEEEauthorblockN{Yin Sun}
\IEEEauthorblockA{\textit{Dept. of ECE} \\
\textit{Auburn University}\\
Auburn, AL, USA \\
yzs0078@auburn.edu}
}

\maketitle

\begin{abstract}
    In this paper, we study the optimal timing for pilot and data transmissions to maximize effective throughput, also known as goodput, over a wireless fading channel. The receiver utilizes the received pilot signal and its Age of Information (AoI), termed the Age of Channel State Information (AoCSI), to estimate the channel state. Based on this estimation, the transmitter selects an appropriate modulation and coding scheme (MCS) to maximize goodput while ensuring compliance with a predefined block error probability constraint. Furthermore, we design an optimal pilot scheduling policy that determines whether to transmit a pilot or data at each time step, with the objective of maximizing the long-term average goodput. This problem involves a non-monotonic AoI metric optimization challenge, as the goodput function is non-monotonic with respect to AoCSI. The numerical results illustrate the performance gains achieved by the proposed policy under various SNR levels and mobility speeds.
\end{abstract}


\section{Introduction}
In wireless communication systems, accurate and timely Channel State Information (CSI) is essential for reliable and efficient data transmission. CSI is obtained through channel estimation, where the transmitter periodically sends pilot signals that allow the receiver to assess the channel state. This information is crucial for optimizing transmission parameters, particularly the selection of the Modulation and Coding Scheme (MCS). However, due to dynamic wireless environments influenced by user mobility and environmental changes, maintaining up-to-date CSI is challenging. To address this, the Age of Channel State Information (AoCSI) \cite{costa2015age}, denoted as $\Delta(t)=t-t_p$, quantifies the freshness of CSI by measuring the time difference between current time $t$ and the generation time of the latest received update $t_p$. As AoCSI increases, the reliability of CSI degrades, potentially leading to suboptimal transmission decisions and increased error rates.

In this paper, we consider a wireless communication system that operates in discrete time slots, with the transmitter deciding at each step whether to send a pilot signal to update the CSI or transmit data using the most recently received CSI. When no pilot is sent, the AoCSI increments, indicating that the available CSI is becoming less reliable. The transmitter then uses this outdated CSI to determine the best MCS for transmitting data. The process begins with channel estimation, where the transmitter relies on the latest available CSI to assess current channel conditions. Based on this estimation, it selects an appropriate MCS, which determines how data is modulated and encoded to ensure reliable transmission while optimizing the data rate. The transmitter aims to maximize overall goodput by selecting the most suitable MCS. It then sends the modulated and encoded data to the receiver, which then demodulates and decodes the received data using the chosen MCS.
In contrast, when a pilot signal is transmitted, the CSI is refreshed, providing the transmitter with accurate and up-to-date channel information.
\begin{figure}[t]
    \centering
\adjincludegraphics[width=1.08\linewidth,trim={0 {.1\height} 0 {.2\height}},clip]{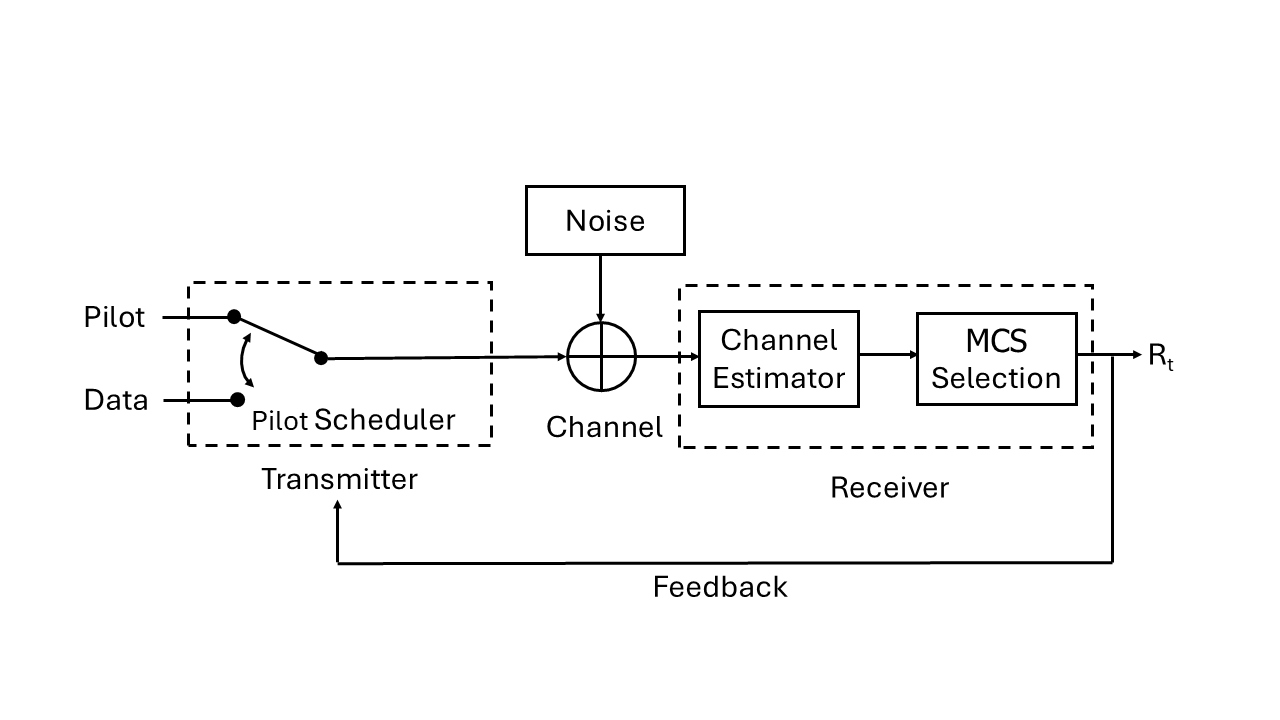}
    \caption{The model of a wireless communication system.}
    \label{diag}
\end{figure}
The key contributions of this work are as follows:
\begin{itemize}
    \item Our framework directly incorporates AoCSI into channel estimation and MCS selection, enabling adaptive modulation and coding strategies based on the freshness of CSI.
    \item We propose a scheduling policy that decides whether to transmit pilot or utilize existing channel information for data transmission. Unlike conventional approaches that assume a monotonically increasing AoI \cite{sun2017update, kadota2018optimizing, bedewy2021optimal, klugel2019aoi, sun2019sampling, kadota2018scheduling}, we recognize that AoCSI can exhibit non-monotonic function for a non-Markovian random channel.
    \item  We formally establish that the optimal pilot scheduling policy is a threshold-based policy which follows a periodic pattern. An index function is used to express the threshold-based policy. In our next step, we will consider the multi-user and multi-antenna case, where the optimal solution may not be periodic.
    \item Numerical results demonstrate the performance gains achieved by the proposed policy under various SNR (Signal to Noise Ratio) levels and mobility speeds, compared to a periodic updating policy with a fixed period. 
\end{itemize}

\section{Related Works}

Age of Information (AoI) has emerged as a fundamental metric for quantifying information freshness in data communication systems since its introduction in \cite{kaul_2012_real}. Early research primarily focused on optimizing average and peak AoI in communication networks \cite{kaul_2012_real, yates2015lazy, sun2017update, kadota2018optimizing}, while more recent studies have explored its applications in real-time systems such as remote estimation \cite{ornee2021sampling, sun2019sampling, chakraborty2025timely}, remote inference \cite{shisher2024timely, shisher2022does, ari2024goal,shisher2023learning}, and control systems \cite{soleymani2019stochastic, klugel2019aoi}. Traditional works \cite{sun2017update, kadota2018optimizing, bedewy2021optimal, klugel2019aoi, sun2019sampling, kadota2018scheduling} largely assumed a monotonic AoI function, where performance degrades predictably as AoI increases. However, recent studies \cite{shisher2021age,shisher2022does,shisher2024timely,shisher2023learning,ari2024goal,chakraborty2025timely} challenged this assumption, demonstrating that information aging can exhibit non-monotonic behavior, for non-Markovian signal. To that end, an information-theoretic interpretation of information freshness was developed in \cite{shisher2021age, shisher2022does,shisher2024timely}. 

AoI-based scheduling policies have also evolved significantly. For Markovian data with monotonic AoI, ``generate-at-will'' policies have been widely studied in \cite{yates2015lazy, sun2017update, sun2019sampling, ornee2021sampling, bedewy2021optimal}, which prioritize transmission of the freshest data. More recent work \cite{shisher2024timely} introduced a ``selection-from-buffer'' policy, which optimally selects updates from a stored buffer to address non-Markovian data processes. While these approaches assume the feasibility of buffering in the context of data transmission, our work differs in that as it focuses on pilot transmissions, where buffering is not feasible. 

AoCSI extends the concept of AoI to wireless channels, quantifying the staleness of CSI and its impact on system performance. The paper, \cite{costa2015age} were among the first to formally define AoCSI, analyzing its effect on wireless links and proposing strategies to mitigate channel aging. Subsequent works \cite{costa2015csit, klein2017staleness} have considered stale CSI at the transmitter, examining its effects on feedback delay and precoding strategies. Later researches \cite{costa2015csit, klein2017staleness, truong2013effects, lipski2024age} focused on finite-state Markov models and assumed a monotonic AoCSI function, where outdated CSI consistently degrades performance in scheduling, power control, and beamforming \cite{lipski2024age}. 
In contrast, our evaluations utilize Jakes’ fast fading model \cite{baddour2005autoregressive} to analyze the impact of AoCSI on channel estimation and MCS selection. As the temporal variation of channel coefficients in Jakes’ model follows a non-Markovian pattern, our numerical evaluation shows that the maximum achievable goodput is a non-monotonic function of AoCSI, which provides new insights into AoCSI-aware communication.

\section{Model and Formulation}

\subsection{System Model}\label{sys_model}

We consider a block fading wireless communication system, as depicted in Figure \ref{diag}. At each time slot 
$t\in\{1,2,\ldots\}$, the transmitter decides to either send a pilot signal or transmit data. We define an indicator function $c(t)\in \{0,1\}$, where $c(t)=1$ indicates that the transmitter is transmitting a pilot at time $t$, and $c(t)=0$ means it is transmitting data. If a pilot, known to both the transmitter and the receiver, is sent at time $t$ with a transmission power of $P_p$, the receiver receives a noisy version of the pilot, given by 
\begin{equation}
y_{\text{pilot},t}=\sqrt{P_p}h_t+n_t,
\end{equation}
where $h_t$ is the channel state at time $t$, assumed to be a circular symmetric complex Gaussian random variable with zero mean,  and $n_t \sim \mathcal{C}\mathcal{N}(0,\sigma_n^2)$ is a circular symmetric complex Gaussian noise with zero mean and $\sigma_n^2$ variance. If a data packet $x_t$ is transmitted at time $t$ with transmission power $P_d$, the received signal is expressed as
\begin{equation}
y_{\text{data},t}=\sqrt{P_d}h_t x_t+n_t,
\end{equation}
where $x_t$ has zero mean and unit variance.

Pilot transmissions consume resources (time, power, and bandwidth) that could otherwise be used for data transmission. In reality, pilot signals are not transmitted at every time slot. When data is transmitted, the receiver has to select the modulation and coding scheme based on the latest received pilot. At time $t$, suppose that the latest pilot signal was received $\Delta(t)$ time slots ago, denoted as $y_{\text{pilot}, t-\Delta(t)}$. Here, $\Delta(t)$ is the \emph{Age of Channel State Information (AoCSI)} \cite{costa2015age}, which represents the time difference between the latest CSI update and the current time $t$. The AoCSI evolves as follows
\begin{equation}\label{aocsi}
    \Delta(t+1) = 
 \begin{cases}
 1, & \text{if} \quad c(t) = 1 \quad (\text{pilot transmitted}), \\
 \Delta(t) + 1, & \text{if} \quad c(t) = 0 \quad (\text{data transmitted}).
 \end{cases}
\end{equation}

The receiver estimates the current channel state $h_t$ using the most recently received pilot $y_{\text{pilot}, t- \Delta(t)}$. Because $h_t$ and $y_{\text{pilot}, t- \Delta(t)}$ are jointly complex Gaussian with zero mean, a linear MMSE channel estimator is employed. The reconstructed channel estimate $\hat h_t$ is given by 
\begin{equation}
    \hat{h}_t=\frac{\sqrt{P_p}\rho_h(\Delta(t))}{P_p\rho_h(0)+\sigma_n^2}y_{\text{pilot}, t- \Delta(t)},
\end{equation}
where $\rho_h(\Delta(t))=\mathbb E [h_t h^*_{t-\Delta(t)}]$ is the auto-covariance function of the channel process $h_t$, assumed to be known at the receiver, and $z^*$ denotes the complex conjugate of $z$. 
The channel state $h_t$ can be expressed as the sum of the channel estimate $\hat{h}_t$ and the estimation error $\bar{h}_t$, i.e.,
\begin{equation}
h_t=\hat{h}_t+\bar{h}_t.
\end{equation}
Given AoCSI $\Delta(t)=\delta$ and the latest received pilot $y_{\text{pilot},t-\Delta(t)}=y$, the variance of channel estimate is 
\begin{align}\label{cov_est1}\nonumber
\text{var}[\hat{h}_t|y_{\text{pilot},t-\delta}=y]\nonumber
&=\mathbb{E}[|\hat{h}_t|^2|y_{\text{pilot},t-\delta}=y],\\
&=\bigg(\frac{\sqrt{P_p}|\rho_h(\delta)|}{P_p\rho_h(0)+\sigma_n^2}\bigg)^2|y|^2.
\end{align}
Due to the orthogonality property of MMSE estimation, $\hat{h}_t$ and $\bar{h}_t$ are independent of each other. 
The signal to interference and noise ratio (SINR) can be obtained as 
\begin{align}\nonumber\label{sinr}
    \eta (\delta,y)&=\frac{\text{var}\left[\sqrt P_d \hat h_t x_t|y_{\text{pilot},t-\delta}=y\right]}{\text{var}\left[\sqrt P_d  \bar{h}_t x_t + n_t|y_{\text{pilot},t-\delta}=y\right]},\\
    &=\frac{P_d\text{var}\Big[ \hat h_t|y\Big]}{P_d\text{var}\Big[\bar{h_t}|y\Big]+\text{var}\Big[ n_t|y \Big]},\\\nonumber
    &=\frac{P_d\bigg(\frac{\sqrt{P_p}|\rho_h(\delta)|}{P_p\rho_h(0)+\sigma_n^2}\bigg)^2|y|^2}{P_d \bigg(\rho_h(0)-\frac{P_p|\rho_h(\delta)|^2}{P_p\rho_h(0)+\sigma_n^2}\bigg)+\sigma_n^2}.
\end{align}

\subsection{MCS Selection}
If data is scheduled for transmission in time slot $t$, i.e., $c(t)=0$, the transmitter selects an optimal modulation and coding scheme (MCS) $R_t$ based on the SINR 
\begin{equation}
    \eta_t = \eta(\Delta(t), y_{\text{pilot},t-\Delta(t)}),
\end{equation}
where $\eta(\cdot,\cdot)$ is defined in \eqref{sinr}.
The transmitter sends a modulated and encoded signal over a block fading channel, while the receiver demodulates the received signal and then performs channel decoding. The decoding error probability, referred to as the Block Error Rate (BLER) and denoted as \( e(\eta_t, R_t) \), depends on the SINR \( \eta_t \) and the selected MCS, where \( R_t \) represents the data rate of the chosen MCS. The objective is to maximize the goodput $R_t[1-e(\eta_t, R_t)]$, while ensuring that the BLER satisfies the constraint $e(\eta_t, R_t) \leq e_\text{max}$. For instance, in LTE data communications, the maximum target BLER is typically set to $e_\text{max}=10\%$ before HARQ retransmissions \cite{piro2010simulating}. Hence, the maximum goodput at time slot $t$ is determined by
\begin{align}
    G_\text{max}(\eta_t) = &\max_{R_t} R_t[1-e(\eta_t, R_t)],\\
&~\text{s.t.}~~ e(\eta_t, R_t) \leq e_\text{max}.
\end{align} 
Given the AoCSI $\Delta(t)=\delta$, the average maximum goodput of time slot $t$ is 
\begin{align}
    r(\delta)=&\mathbb E[G_\text{max}(\eta_t)|\Delta(t)=\delta],\\
    =&\mathbb E\left[G_\text{max}\left(
    \eta\left(\Delta(t),{y_{\text{pilot},t-\Delta(t)}}\right)\right)|\Delta(t)=\delta\right].
\end{align}
 

\subsection{Pilot Scheduling Policy and Problem Formulation}

A pilot scheduling policy, denoted by $\pi =(c(1), c(2), \ldots)$, determines at each time slot whether to transmit a pilot to refresh the CSI or to transmit data based on the latest pilot. We consider AoCSI-based pilot scheduling, where the pilot scheduling policy $\pi$ is determined based on the AoCSI process. Let $\Pi$ be the set of all possible causal scheduling policies, in which every decision $c(t)\in\{0,1\}$ is made based on the current and history AoCSI $\{\Delta(s), s=1,2,\ldots, t\}$.

The objective of pilot scheduling is to develop an optimal scheduling policy that maximizes the long-term average goodput over an infinite time horizon, which is formulated as 
\begin{equation} \label{prob}
    \bar{r}_{opt}=
    \max_{\pi\in \Pi}\liminf_{T\rightarrow \infty} \frac{1}{T}\sum_{t=1}^T \mathbb{E}_{\pi}[(1-c(t))r(\Delta(t))],
\end{equation}
where $\bar{r}_{opt}$ is the optimal objective value of \eqref{prob}. 

\section{Optimal Pilot Scheduling Policy}

Problem \eqref{prob} differs from existing studies on AoI metric optimization in the following aspects: (i) When data is transmitted, i.e., $c(t)=0$, the goodput $[1-c(t)]r(\Delta(t))$ is a non-monotonic function of the AoCSI $\Delta(t)$, which is different from the monotonic AoI metrics studied in, e.g., \cite{sun2017update, kadota2018optimizing, bedewy2021optimal, klugel2019aoi, sun2019sampling, kadota2018scheduling}. (ii) When a pilot transmission occurs, i.e., $c(t)=1$, the goodput $[1-c(t)]r(\Delta(t))$ is zero, which is different from other recent studies \cite{shisher2022does,shisher2024timely,shisher2023learning,ari2024goal,chakraborty2025timely} on non-monotonic AoI metric optimization. To address these challenges, we establish the following Theorem \ref{thm1} to solve \eqref{prob}. 

Define the following index function: 
\begin{align}\label{index}
    \gamma(\delta)=\max_{\tau\in\{1,2,\ldots\}}\frac{1}{\tau}\sum_{k=0}^{\tau-1}\bigg[r(\delta+k)\bigg].
\end{align}
\begin{theorem} \label{thm1}\textit{There exists an optimal solution $\pi^*=(c^*(1),c^*(2),\ldots)$ to problem \eqref{prob}, where}
\begin{align}\label{decision}
    c^*(t)=\begin{cases}
        1, &\textit{if} 
 ~\gamma(\Delta(t))\leq\beta,\\
        0, &\textit{otherwise},
    \end{cases}
\end{align}
\textit{$\gamma(\Delta(t))$ is defined in \eqref{index},}
\textit{the threshold $\beta$ is the unique root of equation \eqref{root}:}
\begin{align}\label{root}
\sum_{\delta=0}^{\tau(\beta)-1} r(\delta)-\beta[\tau(\beta)+1]=0,
\end{align}
\textit{and $\tau(\beta)$ is determined by equation \eqref{waiting}:
\begin{equation}\label{waiting}
\tau(\beta) = \min_{\delta}\{\delta: \gamma(\delta) \leq \beta, \delta=0,1,2,\ldots\}. 
\end{equation}
Moreover, $\beta$ is exactly the optimal objective value of \eqref{prob}, i.e., $ \beta=\bar{r}_\text{opt}$.}
\end{theorem}

\begin{proof}
See Appendix \ref{appendix:A1}. 
\end{proof}

The optimal pilot scheduling policy, as defined in \eqref{decision}, implies that the transmitter decides to transmit a pilot when the index $\gamma(\Delta(t))$, as a function of AoCSI $\Delta(t)$, drops below a threshold $\beta$. This threshold is determined by solving the unique root of \eqref{root} using efficient algorithms such as \textit{Bisection Search} or \textit{Newton’s Method}, as outlined in Algorithms 1-3 of \cite{ornee2021sampling}. Note that because $r(\delta)$ is non-monotonic with respect to the AoCSI $\delta$, $\gamma(\delta)$ can also be non-monotonic in $\delta$. Moreover, \eqref{decision}-\eqref{waiting} have incorporated the fact that the goodput is zero when $c(t)=1$. 
Building upon Theorem \ref{thm1}, we derive the following corollary:


\begin{corollary}\label{coro1}
\textit{The optimal policy $\pi^*$ in Theorem \ref{thm1} is a periodic policy in which one pilot is sent every $(\tau(\beta) + 1)$ time slots.}
\end{corollary}

Leveraging the Bellman optimality equation of \eqref{prob}, Theorem \ref{thm1} and Corollary \ref{coro1} formally establish that the optimal pilot scheduling policy follows a periodic pattern in the single user case. The optimal period $(\tau(\beta) + 1)$ depends on the channel model, carrier frequency, mobility speed, SNR level, and etc. In our future work, we will extend this analysis to a multi-user setting, where the optimal pilot scheduling policy deviates from a simple periodic strategy.
\begin{figure}[!t]
    \centering
\includegraphics[width=1 \linewidth]{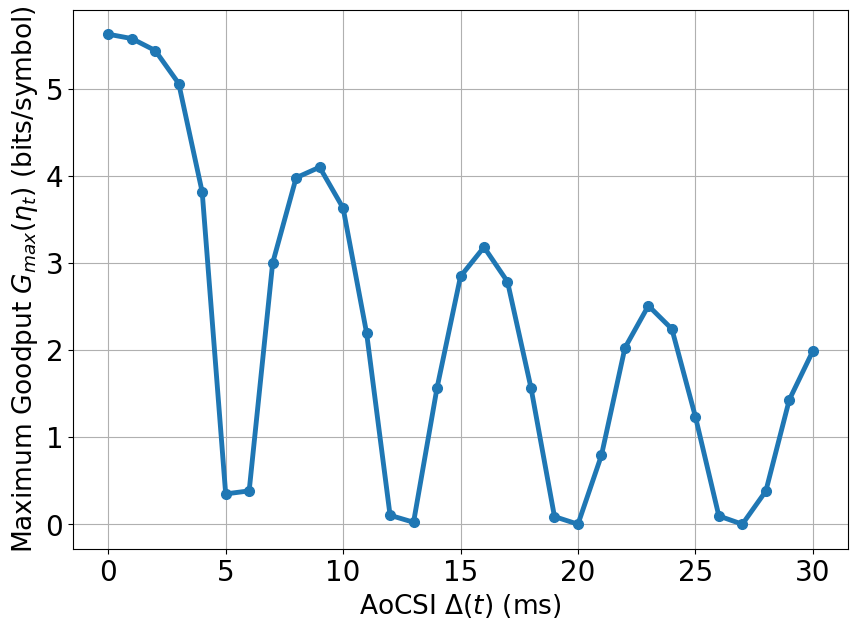}
    \caption{Maximum goodput $G_\text{max}(\eta_t)$ as a non-monotonic function of AoCSI $\Delta(t)$.}
    \label{diag_gp}
\end{figure}
\section{Evaluation}
In this section, we evaluate the performance of our threshold-based scheduling policy in optimizing pilot transmission decisions and the estimator’s performance in MCS selection under various conditions.

For our simulations, we model the wireless channel using a Rayleigh fading process based on Jakes’ model \cite{baddour2005autoregressive}, a widely adopted framework for characterizing time-correlated fading in wireless systems. In this model, the channel coefficients evolute following the autocorrelation function \cite{baddour2005autoregressive}
\begin{equation}\label{auto}
\rho_h(\delta) = J_0(2\pi f_d T_s \delta),
\end{equation}
where $f_d = \frac{v f_c}{c}$ denotes the Doppler frequency, $v$ represents the user velocity, $f_c$ is the carrier frequency, $c$ is the speed of light, $J_0(\cdot)$ is the zeroth-order Bessel function of the first kind, $T_s$ is the sampling period, and $\delta$ corresponds to the time interval influenced by variations in user velocity or CSI update delay. In our simulations, we set the sampling period to $T_s = 1$ ms and the carrier frequency to $f_c = 2.4$ GHz.

We consider 15 distinct MCS levels, corresponding to Channel Quality Indicator (CQI) values from 1 to 15 in the LTE standard \cite{3gpp_ts_36_213}. The modulation schemes and code rates adhere to the LTE standard. The BLER values are drawn from the precomputed BLER-SINR curves for Additive White Gaussian Noise (AWGN) channels \cite{piro2010simulating}, obtained using an LTE link-level simulator.



\subsection{Numerical Results for Estimator}



We evaluate the relationship between the AoCSI $\Delta(t)$ and the maximum achievable goodput $\gamma(\Delta(t))$, showing the non-monotonic behavior of goodput as AoCSI increases. For this analysis, we consider a user mobility of $v=15$ mph and $20$ dB SNR level. We see from Figure \ref{diag_gp} that initially goodput is high when the CSI is fresh, but as AoCSI increases, goodput rapidly declines due to the growing in channel state estimation. Instead of a continuously decreasing trend, the plot exhibits a fluctuating pattern, where goodput temporarily recovers before declining again. This behavior arises from the fluctuating pattern of the zeroth-order Bessel function $J_0(\cdot)$. 

\begin{figure}[!t]
    \centering
\includegraphics[width=1 \linewidth]{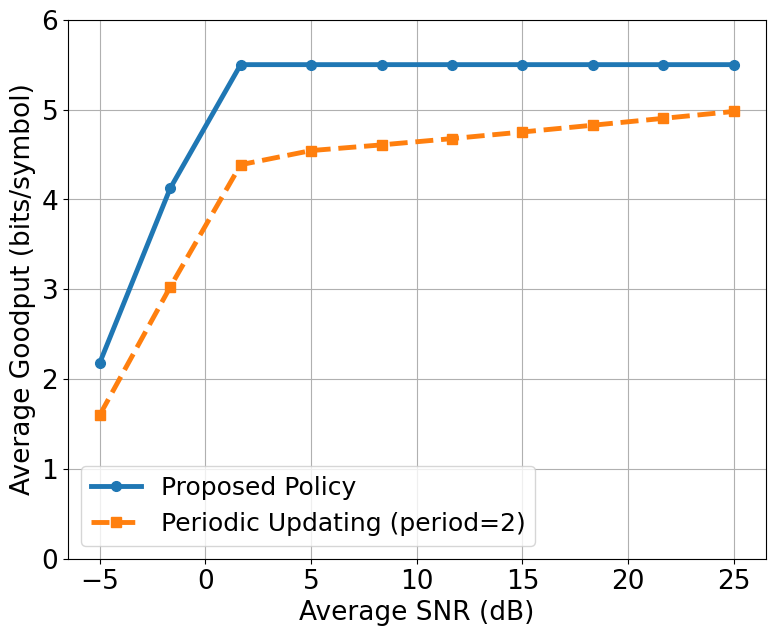}
    \caption{Performance comparison between proposed policy and periodic updating policy with update period as 2 time slots, with varying SNR values between $-5$ and $25$ dB.}
    \label{pol_snr}
\end{figure}

\subsection{Numerical Results for Scheduler}

We evaluate the performance of our proposed scheduling
policy in maximizing the overall goodput. We compare the following two scheduling policies for transmission decisions:
\begin{itemize}
    \item Periodic Updating: In this policy the transmitter decides to transmit pilot after waiting for every $2$ time slots.
    \item Optimal Threshold-Based Policy: The policy we propose in Theorem 1.
\end{itemize}
Figure \ref{pol_snr} shows the comparison in average goodput performance of the proposed policy with the periodic updating policy as a function of average SNR. The SNR is varied from -5 dB to 25 dB, with a fixed user mobility of $v=15$ mph. As expected, both policies exhibit an increasing trend in goodput with increasing SNR and eventually stabilize at high SNR values. However, the proposed policy consistently outperforms the periodic updating policy across all SNR levels and reaches the highest achievable average goodput faster.
\begin{figure}[!t]
    \centering
\includegraphics[width=1 \linewidth]{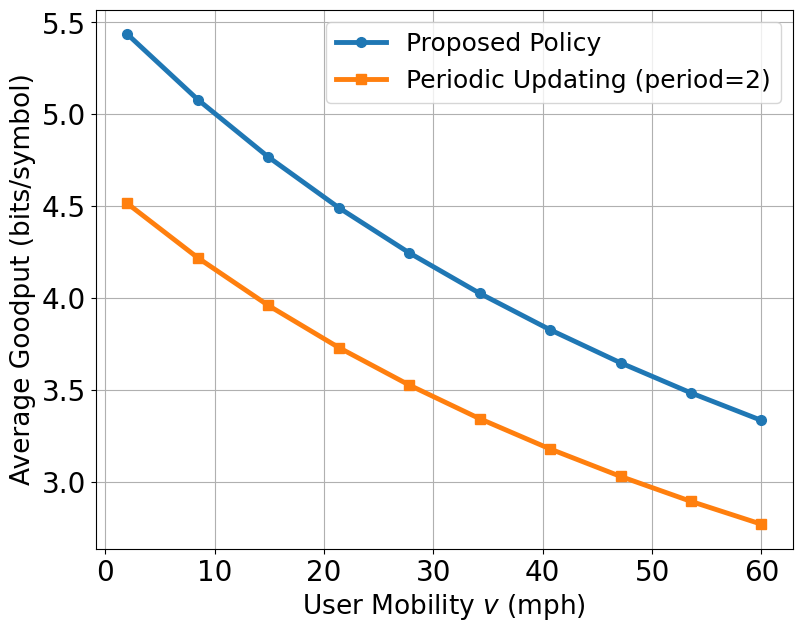}
    \caption{Performance comparison between the proposed policy and the periodic updating policy with update period as 2 time slots and varying user speed $v$ from $20-70$ mph.}
    \label{pol_m}
\end{figure}

Next, we compare the average goodput performance of the two policies defined earlier with increasing user mobility, as shown in Figure \ref{pol_m}. The user mobility $v$ is varied between 2 mph and 60 mph, covering a range of real-world scenarios from pedestrian speeds to vehicular motion. As mobility increases, both policies show a decline in goodput due to faster channel variations degrading CSI accuracy. However, the proposed policy consistently outperforms periodic updating.


\section{Conclusion}
This paper investigates the role of AoCSI in optimizing wireless communication systems, focusing on its impact on MCS selection. We analyze the relationship between CSI timeliness, channel estimation, MCS selection and system performance, proposing an optimal pilot scheduling policy to maximize long-term goodput. Numerical results demonstrated that the proposed threshold policy effectively addresses the challenges posed by outdated CSI, ensuring reliable and efficient data transmission. Our future work will extend the framework to multi-user and multi-antenna systems, where additional system complexities may lead to non-periodic pilot scheduling. 

\section*{Acknowledgment}
The authors thank Sam Chamoun for his assistance with the simulations. 
The second author is also grateful to Md Kamran Chowdhury Shisher for earlier discussions on this topic.

\printbibliography


\begin{appendices}
\section{Proof of Theorem 1}\label{appendix:A1}
The proof follows a similar idea to \cite[Theorem 4]{shisher2024timely}, but the steps differ due to the distinct problem formulation \eqref{prob} for pilot scheduling. 

Problem \eqref{prob} is an infinite-horizon time-averaged MDP problem, with the following components:
\begin{itemize}
    \item State: At time slot $t$, the state of the MDP is defined by the AoCSI $\Delta(t)$.
    \item Action: The action of the scheduler is either transmit a pilot or not, represented by the indicator function $c(t)\in\{0,1\}$.
    \item Reward: The reward at time $t$ is given by  $(1-c(t))r(\Delta(t))$.
    \item State Transition: The AoCSI evolves following \eqref{aocsi}.
\end{itemize}
The Bellman optimality equation of \eqref{prob} is given by
\begin{align}\label{bell1}
    &h(\Delta(t))=\\\nonumber
    &\max_{c \in \{0,1\}} \Big\{[(1-c)r(\Delta(t))-\bar{r}_\text{opt}]+\mathbb{E}[h(\Delta(t+1))]\Big\},
\end{align}
for all time $t=0, 1,2,\ldots$, where $h(\delta)$ is the relative value function of \eqref{prob}.

Given $\Delta(t)=\delta$ and suppose that the system chooses to wait (i.e. transmits data) for $\tau$ time slots before sending a pilot in the $(\tau+1)$-th slot, the Bellman optimality equation \eqref{bell1} is equivalently expressed as
\begin{align}\label{bell2}\nonumber
    &h(\delta)=\\
    &\max_{\tau\in\{0,1,2,\ldots\}}\Big[\sum_{k=0}^{\tau-1}r(\delta+k)-(\tau+1)\bar{r}_{opt}+h(1)\Big], \delta=1,2,\ldots
\end{align}

Let $\tau^*(\delta)$ denote the optimal waiting time that maximizes the right hand side of \eqref{bell2}. The optimal decision is to transmit a pilot immediately at state $\delta$, i.e. $\tau^*(\delta)=0$, if 
\begin{align}\label{pr6}\nonumber
    \max_{\tau\in\{1,2,\ldots\}}\bigg[\sum_{k=0}^{\tau-1}r(\delta+k)-(\tau+1)\bar{r}_{opt}\bigg]\\\leq
    \sum_{k=0}^{-1}r(\delta+k)-(0+1)\bar{r}_{opt},
\end{align}
Simplifying \eqref{pr6}, we get
\begin{align}\label{compare}
    \max_{\tau\in\{1,2,\ldots\}}\bigg[\sum_{k=0}^{\tau-1}r(\delta+k)-(\tau+1)\bar{r}_{opt}\bigg]\leq -\bar{r}_\text{opt}.
\end{align}
Rearranging \eqref{compare}, we write
\begin{equation}\label{pr1}
    \max_{\tau\in\{1,2,\ldots\}}~~\bigg[\sum_{k=0}^{\tau-1}r(\delta+k)- \tau\bar{r}_{opt}\bigg]\leq 0.
\end{equation}
Similar to Lemma 2 in \cite{sun2019sampling}, inequality \eqref{pr1} holds, if and only if
\begin{equation}\label{pr2}
    \max_{\tau\in\{1,2,\ldots\}}\frac{1}{\tau}\bigg[\sum_{k=0}^{\tau-1}r(\delta+k)\bigg]\leq \bar{r}_{opt}.
\end{equation}
Using \eqref{index}, we can rewrite \eqref{pr2} as
\begin{equation}
    \gamma(\delta)\leq \bar{r}_{opt}.
\end{equation}
In other words, $\tau^*(\delta)=0$, if $\gamma(\delta)\leq \bar{r}_{opt}$.
Hence, the optimal waiting time $\tau^*(\delta)$ that solves \eqref{bell2} is determined by
\begin{equation}
    \tau^*(\delta)=\min_k\{k\geq 0: \gamma(\delta+k)\leq \bar{r}_\text{opt}\}.
\end{equation}


Next, to determine the optimal objective value $\bar{r}_{opt}$, we set $\delta=1$ in \eqref{bell2}, resulting in
\begin{equation}\label{pr3}
h(1)=\sum_{k=0}^{\tau-1}r(k+1)-(\tau+1)\bar{r}_{opt}+h(1).
\end{equation}
From \eqref{pr3}, we get
\begin{equation}\label{pr4}
    \sum_{k=0}^{\tau-1}r(k+1)-(\tau_{}+1)\bar{r}_{opt}=0.
\end{equation}
Rearranging \eqref{pr4}, we get 
\begin{equation}\label{pr5}
    \bar{r}_{opt}=\frac{\sum_{k=0}^{\tau-1}r(k+1)}{\tau_{}+1}.
\end{equation}

Let $\tau(\beta)$ be the waiting time under a threshold $\beta$, we rewrite \eqref{pr5} as
\begin{equation}
    \bar{r}_{opt}=\frac{\sum_{t=0}^{\tau(\beta)-1}r(\Delta(t))}{\tau(\beta)+1},
\end{equation}

Finally, in order to find the threshold $\beta$, using the proof techniques for Lemma 2 in \cite{ornee2021sampling}, we say the following equation
\begin{align}
\sum_{t=0}^{\tau(\beta)-1} r(\Delta(t))-\beta[\tau(\beta)+1]=0
\end{align}
has a unique root $\beta$.
\end{appendices}

\end{document}